\documentclass[conference,letterpaper]{IEEEtran}
\usepackage[utf8]{inputenc} 
\usepackage[T1]{fontenc}
\usepackage{url}
\usepackage{ifthen}
\usepackage{cite}
\usepackage[cmex10]{amsmath} 
\usepackage{cite}
\usepackage{amssymb,amsfonts}
\usepackage{algorithmic}
\usepackage{graphicx}
\usepackage{textcomp}
\usepackage{xcolor}
\usepackage{epstopdf}
\usepackage{diagbox}
\usepackage{amsthm}
\usepackage{float}
\usepackage{subfigure}
\newtheorem{Le}{Lemma}
\newtheorem{theo}{Theorem}                             

\begin{document}
\title{On the Weight Spectrum of Pre-Transformed Polar Codes} 

%%% Many authors with many affiliations:
 \author{%
   \IEEEauthorblockN{Yuan Li\IEEEauthorrefmark{1}\IEEEauthorrefmark{2}\IEEEauthorrefmark{3},
                     Huazi Zhang\IEEEauthorrefmark{1},
                     Rong Li\IEEEauthorrefmark{1},
                     Jun Wang\IEEEauthorrefmark{1},
                     Guiying Yan\IEEEauthorrefmark{2}\IEEEauthorrefmark{3},
                     and Zhiming Ma\IEEEauthorrefmark{2}\IEEEauthorrefmark{3}}
   \IEEEauthorblockA{\IEEEauthorrefmark{1}%
                     Huawei Technologies Co. Ltd.}
  \IEEEauthorblockA{\IEEEauthorrefmark{2}%
                     University of Chinese Academy of Sciences}
   \IEEEauthorblockA{\IEEEauthorrefmark{3}%
                     Academy of Mathematics and Systems Science, CAS }
    Email: liyuan181@mails.ucas.ac.cn, \{zhanghuazi, lirongone.li, justin.wangjun\}@huawei.com,\\
           yangy@amss.ac.cn, mazm@amt.ac.cn 

 }

\maketitle

%%%%%%
%% Abstract: 
%% If your paper is eligible for the student paper award, please add
%% the comment "THIS PAPER IS ELIGIBLE FOR THE STUDENT PAPER
%% AWARD." as a first line in the abstract. 
%% For the final version of the accepted paper, please do not forget
%% to remove this comment!
%%
\begin{abstract}
Polar codes are the first class of channel codes achieving the symmetric capacity of the binary-input discrete memoryless channels (B-DMC) with efficient encoding and decoding algorithms. But the weight spectrum of polar codes is relatively poor compared to Reed-Muller (RM) codes, which degrades their maximum-likehood (ML) performance. Pre-transformation with an upper-triangular matrix (including cyclic redundancy check (CRC), parity-check (PC) and polarization-adjusted convolutional (PAC) codes), improves weight spectrum while retaining polarization. In this paper, the weight spectrum of upper-triangular pre-transformed polar codes is mathematically analyzed. In particular, we focus on calculating the number of low-weight codewords due to their impact on error-correction performance. Simulation results verify the accuracy of the analysis.
\end{abstract}

\section{Introduction}

Polar codes \cite{b1}, invented by Ar{\i}kan, are a great break through in coding theory. As code length $N = 2^n$ approaches infinity, the synthesized channels become either noiseless or pure-noise, and the fraction of the noiseless channels approaches channel capacity. Thanks to channel polarization, efficient  successive cancellation (SC) decoding algorithm can be implemented with a complexity of $O(NlogN)$. However, the performance of polar codes under SC decoding is poor at short to moderate block lengths.

In \cite{b2}, a successive cancellation list (SCL) decoding algorithm was proposed. As the list size $L$ increases, the performance of SCL decoding approaches that of ML decoding. But the ML performance of polar codes is still inferior due to low minimum distance. Consequently, concatenation of polar codes with CRC \cite{b3} and PC \cite{b4} were proposed to improve weight spectrum.
 Recently, Ar{\i}kan proposed polarization-adjusted convolutional (PAC) codes \cite{b5}, which is shown to approach binary input additive white Gaussian noise (BIAWGN) dispersion bound \cite{b7} under large list decoding\cite{b8}.

CRC-Aided (CA) polar, PC-polar, and PAC codes can be viewed as pre-transformed polar codes with upper-triangular transformation matrices\cite{b8}. In\cite{b9}, it is proved that any pre-transformation with an upper-triangular matrix does not reduce the minimum Hamming weight, and a properly designed pre-transformation can reduce the number of minimum-weight codewords. In this paper, we propose an efficient method to calculate the average weight spectrum of pre-transformed polar codes. Moreover, the method holds for arbitrary information sub-channel selection criteria, thus covers polar codes, RM codes and is not constrained by "partial order"\cite{b6}. Our results confirm that the pre-transformation with an upper-triangular matrix can reduce the number of minimum-weight codewords significantly when the information set is properly chosen. In the meantime, it enhances error-correcting performance of SCL decoding.

In section II, we review polar codes and pre-transformed polar codes. In section III we propose a formula to calculate the average weight spectrum of pre-transformation polar codes. In section IV the simulation results are presented to verify the accuracy of the formula. Finally we draw some conclusions in section V.

\section{Background}

\subsection{Polar Code}

Given a B-DMC $W: \{ 0,1 \} \rightarrow \mathcal{Y}$,  the channel transition
probabilities are defined as $W(y|x)$ where $y \in \mathcal{Y} , x \in \{ 0,1 \}$. $W$ is said to be symmetric if there is a permutation $\pi$, such that $ \forall$ $y \in \mathcal{Y}$, $W(y|1)=W(\pi(y)|0)$ and $\pi^2 = id$.

Then the symmetric capacity and the Bhattacharyya parameter of  $W$ are defined as
\begin{equation}
I(W) \triangleq \sum_{y \in \mathcal{Y}} \sum_{x \in \mathcal{X}} \frac{1}{2} W(y \mid x) \log \frac{W(y \mid x)}{\frac{1}{2} W(y \mid 0)+\frac{1}{2} W(y \mid 1)},
\end{equation}
and
\begin{equation}
Z(W) \triangleq \sum_{y \in \mathcal{Y}} \sqrt{W(y \mid 0) W(y \mid 1)}.
\end{equation}

Let
$
F=\left[\begin{array}{ll}
1 & 0 \\
1 & 1
\end{array}\right]
$, $N=2^m$, and $F_N=F^{\otimes m}$. Starting from $N = 2^m$ independent channels $W$, we obtain $N$ polarized channels $W_N^{(i)}$, after channel combining and splitting operations \cite{b1}, where
\begin{equation}
W_{N}\left(y_{1}^{N}|u_1^N \right) \triangleq W^N \left(y_1^N|u_1^NF_N\right),
\end{equation}
\begin{equation}
W_{N}^{(i)}\left(y_{1}^{N}, u_{1}^{i-1} \mid u_{i}\right) \triangleq \sum_{u_{i+1}^{N} \in \mathcal{X}^{N-i}} \frac{1}{2^{N-1}} W_{N}\left(y_{1}^{N} \mid u_{1}^{N}\right).
\end{equation}

Polar codes can be constructed by selecting the indices of $K$ information sub-channels, denoted by the information set $\mathcal{A} = \left\{ I_1,I_2,\dots,I_K \right\}$. The optimal sub-channel selection criterion for SC decoding is reliability, i.e., selecting the $K$ most reliable sub-channel as information set. Density evolution (DE) algorithm\cite{b10}, Gaussian approximation (GA) algorithm\cite{b11} and the channel-independent polarization weight (PW) construction algorithm\cite{b12} are efficient methods to find reliable sub-channels. The optimal sub-channel selection criterion for SCL decoding is still an open problem. Some heuristic approaches cosider both reliability and row weight to improve minimum code distance. 

After determining the information set $\mathcal{A}$, the complement set $\mathcal{A}^c$ is called the frozen set. Let 
$u_1^N=(u_1,u_2, \dots , u_N)$ be the bit sequence to be encoded. The information bits are inserted into $u_{\mathcal{A}}$,
and all zeros are filled into $u_{\mathcal{A}^c}$. Then the codeword $x_1^N$ is obtained by $x_1^N=u_1^NF_N$.

\subsection{Weight Spectrum of Polar Codes}
There are several prior works to obtain the weight spectrum of polar codes. In \cite{b15}, the authors use SCL decoding with a large list size to decode an all-zeros codeword. Codewords within the list are enumerated to estimate the number of low-weight codewords. In \cite{r1}, this approach is improved in term of memory usage. The above methods only obtain partial weight spectrum. In \cite{r2}\cite{r3}, probabilistic computation methods are proposed to estimate the weight spectrum of polar codes.

\subsection{Weight Spectrum of Polar Cosets}
As in \cite{r4}, let $u_1^{i-1} \in \{0,1\}^{i-1}$, $u_i \in \{0,1\}$, define the polar coset $C_N^{(i)}\left(u_1^{i-1},u_i \right)$ as
$$C_N^{(i)}\left(u_1^{i-1},u_i \right)=\left\{ (u_1^i,u')F_N|u' \in \{0,1\}^{n-i} \right\}.$$

In \cite{r5}\cite{r6}, recursive formulas are proposed to efficiently compute the weight spectrum of $C_N^{(i)}\left(0_1^{i-1},1 \right)$. The weight spectrum of $C_N^{(i)}\left(0_1^{i-1},1 \right)$ is tightly associated with the performance of SC decoding, our analysis of average weight spectrum of pre-transformed polar codes is based on the polar coset spectrum as well.

\subsection{Pre-Transformed Polar Codes}

\begin{equation*}
T
=\begin{bmatrix}
1  &  T_{12}  & \cdots\ &T_{1N}\\
0  &  1  & \cdots\ & T_{2N}\\
 \vdots   & \vdots & \ddots  & \vdots  \\
 0 & 0  & \cdots\ & 1\\
\end{bmatrix}
\end{equation*}

The above non-degenerate upper-triangular pre-transformation matrix $T$ has all ones on the main diagonal. Let $G_N = TF_N$ and $u_{\mathcal{A}^c}=\textbf{0}$, the codeword of the pre-transformed polar codes is given by $x_1^N=u_1^NG_N=u_1^NTF_N$. In original polar codes, the frozen bits are fixed to be zeros. While in pre-transformed polar codes, the frozen bits are linear combination of previous information bits.

\section{Average Code Spectrum
Analysis}

In this section, we propose a formula to compute the average weight spectrum of the pre-transformed polar codes, with focus on the number of low-weight codewords. The average number assumes that $T_{ij}$, $1\leq i < j \leq N$ are $i.i.d.$  $Bernoulli(\frac{1}{2})$ $r.v.$.

\subsection{Notations and Definitions}

$f_N^{(i)}$ is the $i$-$th$ row vector of $F_N$, and $g_N^{(i)}$ is the $i$-$th$ row vector of $G_N$. The number of codewords with Hamming weight $d$ of the pre-transformed polar codes is denoted by $N_d(\mathcal{A},T)$. The minimum distance of  polar/RM codes and the pre-transformed codes are denoted by $d_{min}(\mathcal{A})$ and $d_{min}(\mathcal{A},T)$, respectively. The number of minimum-weight codewords of polar/RM codes and the pre-transformed codes are denoted by $N_{min}(\mathcal{A})$ and $N_{min}(\mathcal{A},T)$, respectively.

\subsection{Code Spectrum Analysis}

The expected number of codewords with Hamming weight $d$ is 

\begin{align}
\label{eq:1}
&E\left[ N_d(\mathcal{A},T)\right] \notag \\ 
&= \sum_{u_{I_1},u_{I_2},\dots,u_{I_K}\in \{0,1\}^K} P\left(w \left(\oplus \sum_{i=1}^K u_{I_i} g_N^{(I_i)} \right)=d \right) \notag \\
&= \sum_{j=1}^K \sum_{\mbox{\tiny$\begin{array}{c}
u_1,\dots,u_{I_{j-1}}=0\\
u_{I_j}=1\\
u_{I_{j+1}},\dots,u_{I_K} \in \left\{0,1\right\}^{K-j}\end{array}$}} \notag \\
&P \left(w\left(g_N^{(I_j)}\oplus \sum_{i=j+1}^K u_{I_i} g_N^{(I_i)}\right)=d\right).
\end{align}
The expectation is with respect to $T$, so $E\left[ N_d(\mathcal{A},T)\right]$ is a function of $\mathcal{A}$.

\begin{Le}
${\forall}$ $u_{I_{j+1}},\dots,u_{I_K} \in \left\{0,1\right\}^{K-j}$,
$$P \left(w\left(g_N^{(I_j)}\oplus \sum_{i=j+1}^K u_{I_i} g_N^{(I_i)}\right)=d \right)=P \left(w\left(g_N^{(I_j)}\right)=d \right).$$
\end{Le}

\begin{proof}According to the pre-transformation matrix,
\begin{equation*}
g_N^{(I_j)}=f_N^{(I_j)}\oplus \sum_{i=I_j+1}^N T_{I_ji} f_N^{(i)},
\end{equation*}
$$g_N^{(I_j)}\oplus \sum_{i=j+1}^K u_{I_i} g_N^{(I_i)}=f_N^{(I_j)}\oplus \sum_{i=I_j+1}^N T'_{I_ji} f_N^{(i)}.$$
And
\begin{equation*}    T'_{I_ji} \triangleq
 \begin{cases}
    \sum_{I_k < i,u_{I_k}=1}T_{I_ki}   & i \notin \left[I_{j+1},\dots,I_K \right] \\
    \sum_{I_k < i,u_{I_k}=1}T_{I_ki} \oplus u_i & i \in \left[I_{j+1},\dots,I_K \right].
\end{cases}
\end{equation*}

It is straightforward to see that when $T_{I_ji}$ are $i.i.d.$ $Bernoulli(\frac{1}{2})$ $r.v.$, $T'_{I_ji}$ are $i.i.d.$ $Bernoulli(\frac{1}{2})$ $r.v.$ as well.

As a result, $g_N^{(I_j)}$ and $g_N^{(I_j)}\oplus \sum\limits_{i=j+1}^K u_{I_i} g_N^{(I_i)}$ follow the same distribution too, ${\forall}$ $u_{I_{j+1}},\dots,u_{I_K} \in \left\{0,1\right\}^{K-j}$.
Therefore \textbf{Lemma 1} holds.
\end{proof}

\begin{Le}
If $w(f^{(I_j)}_N) > d$, $P \left(w(g_N^{(I_j)})=d \right)=0$.
\end{Le}

\begin{proof}
Recall that
$$g_N^{(I_j)}=f_N^{(I_j)}\oplus \sum_{i=I_j+1}^N T_{I_ji} f_N^{(i)}.$$
According to \cite[Corollary 1]{b9},
$$w\left(g_N^{(I_j)}\right) \geq w\left(f_N^{(I_j)} \right)>d,$$
therefore
$$P \left(w\left(g_N^{(I_j)}\right)=d \right)=0.$$
\end{proof}

According to \textbf{Lemma 1} and \textbf{Lemma 2}, (\ref{eq:1}) can be further simplified to

\begin{align}
\label{eq:2}
E\left[ N_d(\mathcal{A},T)\right] = \sum_{\mbox{\tiny$\begin{array}{c}
1 \leq j \leq K\\
w(f_{I_j}) \leq d\end{array}$}}2^{K-j}P \left(w\left(g_N^{(I_j)}\right)=d \right).
\end{align}

Let $P(m,i,d) \triangleq P\left(w\left(g_{2^m}^{(i)}\right)=d\right)$ , (\ref{eq:2}) can be rewritten as

\begin{align}
\label{eq:4}
E\left[ N_d(\mathcal{A},T)\right] = \sum_{\mbox{\tiny$\begin{array}{c}
1 \leq j \leq K\\
w(f_{I_j}) \leq d\end{array}$}}2^{K-j}P \left(m,I_j,d \right).
\end{align}

In particular, let $P(m,i) \triangleq P\left(w\left(g_{2^m}^{(i)}\right)=w\left(f_{2^m}^{(i)}\right)\right)$. So if $d=d_{min}$, (\ref{eq:2}) can be rewritten as\\
$E\left[ N_{min}(\mathcal{A},T)\right] =$
\begin{align}
\label{eq:5}
\sum_{\mbox{\tiny$\begin{array}{c}
1 \leq j \leq K\\
w(f_{I_j})=d_{min}(\mathcal{A})\end{array}$}}2^{K-j}P \left(m,I_j\right).
\end{align}

Let $A_{d}$ denote the number of codewords in $C_N^{(i)}\left(0_1^{i-1},1 \right)$ with Hamming weight $d$. Clearly, $2^{N-i}P(m,i)=A_{w\left(f_N^{(i)}\right)}$, $2^{N-i}P(m,i,d)=A_d$.  
In \cite{r5}\cite{r6}, the authors propose recursive formulas to calculate the weight spectrum of polar cosets. 

In \textbf{Theorem 1} and \textbf{Theorem 2}, we investigate the recursive fomulas for $P(m,i)$ and $P(m,i,d)$, which are similar to the formula in \cite{r6}. But instead of polar cosets, we are interested in the pre-transformed polar codes. For the completeness of the paper, the proofs are in the appendix.   

\begin{theo}
\begin{align}
\label{eq:8}
P(m,i)=
\begin{cases}
\frac{2^{w(f_{2^m}^{(i)})}}{2^{2^{m-1}}}P(m-1,i)& 1 \leq i \leq 2^{m-1}\\
P(m-1,i-2^{m-1})& 2^{m-1} < i \leq 2^m,
\end{cases}
\end{align}
with the boundary conditions $P(1,1)=P(1,2)=1$.
\end{theo}

With (\ref{eq:5}) and (\ref{eq:8}), we can recursively calculate the average number of minimum-weight codewords. We are also interested in other low-weight codewords on the weight spectrum, since together they determine the ML performance at high SNR. The problem boils down to evaluating the more general formula of $P(m,i,d)$. As we will see in \textbf{Theorem 2}, the average weight spectrum can be calculated efficiently in the same recursive manner especially for codewords with small Hamming weight.

\begin{theo}
If $1 \leq i \leq 2^{m-1}$
\begin{align}
\label{eq:r1}
P(m,i,d)=\sum\limits_{\mbox{\tiny$\begin{array}{c}
d'=w\left(f_{2^m}^{(i)}\right)\\
d-d' is \ even \end{array}$}}^{d}P(m-1,i,d')\frac{2^{d'}\begin{pmatrix} 2^{m-1}-d' \\ \frac{d-d'}{2} \end{pmatrix}}{2^{2^{m-1}}}.
\end{align}
If \ $2^{m-1} < i \leq 2^m$
\begin{align}
\label{eq:r2}
P(m,i,d)=\begin{cases}P(m-1,i-2^{m-1},d/2) &d \ is \ even\\
0  &d \ is\ odd,
\end{cases}
\end{align}
with the boundary conditions $P(1,1,1)=P(1,2,2)=1$. And 
\begin{align}
\label{eq:6}
\begin{cases}
P(m,1,d)=0,&\ if \ d \ is \ even  \\ 
P(m,i,d)=0,&\ if \ i > 1 \ and \ d \ is \ odd.
\end{cases}
\end{align}
\end{theo}

\subsection{Complexity Analysis}
In this section, we consider the computational complexity of average weight spectrum of pre-transformed polar codes.

According to \textbf{Theorem 1} and (\ref{eq:5}), the computational complexity of average number of minimum-weight codewords is $O(NlogN)$.

Let $\chi(N)$ denote the worst case complexity of computing the average spectrum of pre-transformed polar codes with code length $N$. According to (\ref{eq:r1}) and (\ref{eq:r2}), $O(N)$ operations are required for computing each $P(m,i,d)$, after the computation of average spectrum with code length $\frac{N}{2}$. So the computation complexity of $P(m,i,d)$ over $1 \leq i \leq N$, $0 \leq d \leq N$ is $O(N^3)$. At last, (\ref{eq:4}) requires no more than $N$ calculation. In short, $\chi(N) \leq \chi(N/2) + O(N^3)$. Consequently, $\chi(N) = O(N^3)$.

\section{Simulation}

In this section, we verify the correctness of the recursive formula through simulations. In particular, we employ the "large list decoding" method described in \cite{b15} to collect low-weight codewords.  Transmit all-zero codeword without noise, and use list decoding to decode the channel output. With sufficiently large list size $L$, the decoder collects all the low-weight codewords. At first, we randomly generate one thousand pre-transfom matrices for RM(128, 64), and set $L=5\times 10^3$ to count the number of minimum-weight codewords for each matrix, and obtain their average $N_{min}$. The result is shown in Fig.\;\ref{fig1}: $d_{min}=16$, $N_{min}^{simulation}=2768.1$, $N_{min}^{recursion}=2766.9$.

To show that our recursive formula is applicable for any sub-channel selection criterion we also construct polar code(128, 64) by the PW algorithm \cite{b12}. The simulation result is shown in Fig.\;\ref{fig2}: $d_{min}=8$, $N_{min}^{simulation}=272.64$, $N_{min}^{recursion}=272$.

Our recursive formula is also applicable for longer codes. We set $L=5\times 10^4$ to count the number of minimum-weight codewords for one thousand pre-transformed RM(512, 256). The result is shown in Fig.\;\ref{fig6}: $d_{min}=32$, $N_{min}^{simulation}=1.5933\times 10^4$, $N_{min}^{recursion}=1.5936\times 10^4$.

As seen, the recursively calculated minimum-weight codeword numbers are very close to ones obtained through simulation. Furthermore, they show that the variance of the number of minimum-weight codeword is small. 
\begin{figure}[htbp]
\centerline{\includegraphics[width = .5\textwidth]{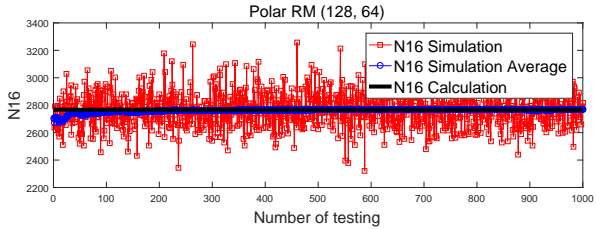}}
\caption{RM(128,64), the black solid line is calculated, and the blue solid line is obtained from simulation.}
\label{fig1}
\end{figure}

\begin{figure}[htbp]
\centerline{\includegraphics[width = .5\textwidth]{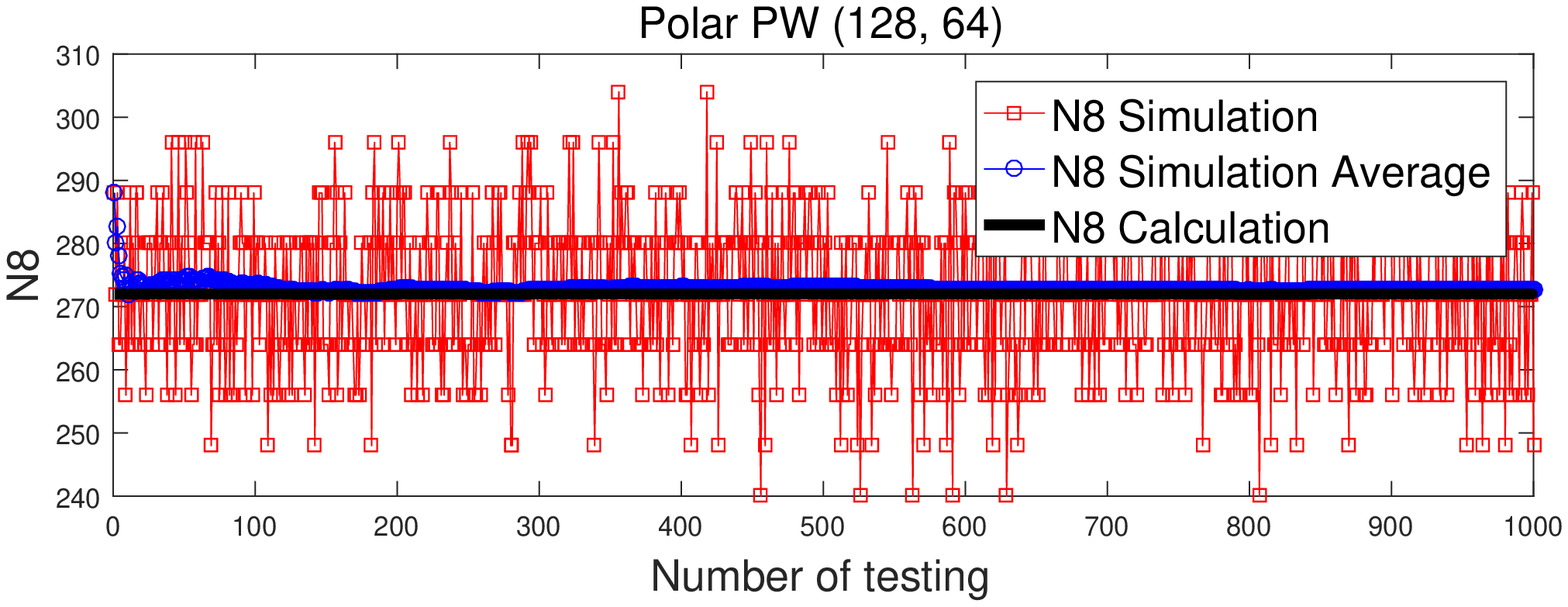}}
\caption{ PW(128,64), the black solid line is calculated, and the blue solid line is obtained from simulation.}
\label{fig2}
\end{figure}

\begin{figure}[htbp]
\centerline{\includegraphics[width = .5\textwidth]{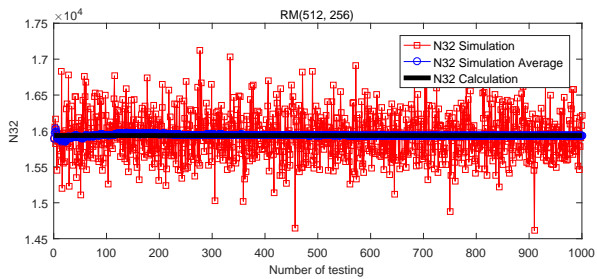}}
\caption{RM(512,256), the black solid line is calculated, and the blue solid line is obtained from simulation.}
\label{fig6}
\end{figure}

In Table.\;\ref{tab1}, we display the number of minimum codewords of the original RM/polar codes, and the average number is recursively calculated. It is shown that pre-transforming significantly reduces the number of minimum-weight code words, especially in RM(128, 64). The significant improvement of weight spectrum after pre-transformation explains why the CA-polar, PC-polar, and PAC codes outperform the original polar codes under list decoding with large list size.

\begin{table}[htbp]
\caption{comparsion between original polar codes and Pre-transfomed polar codes}
\begin{center}
\begin{tabular}{|c|c|c|c|}
\hline
\ \ &\multicolumn{3}{|c|}{\textbf{Minimum-weight codewords}} \\
\cline{2-4}
\ \  & $d_{min}$ &\textit{Original}& \textit{Pre-trasformed} \\
\hline
RM(128,64)&16&94488 &2767  \\
\hline
PW(128,64)&8&304 &272  \\
\hline
\end{tabular}
\label{tab1}
\end{center}
\end{table}

The improvement can be observed under different code lengths and rates, as we can see from Fig.\;\ref{fig3}. In all cases, pre-transformation reduces the number of minimum codewords significantly.

\begin{figure}[htbp]
\centerline{\includegraphics[width = .5\textwidth]{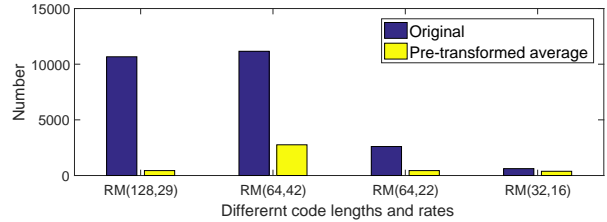}}
\caption{comparsion between original polar codes and pre-transformed polar codes under different code lengths and rates}
\label{fig3}
\end{figure}

In addition to minimum-weight codewords, we also simulate to verify the accuracy of the formula for other low-weight codewords. The simulation results are shown in Table.\;\ref{tab2} for RM(128, 64) and PW(128, 64) respectively, where
$N^{sim}$ is the simulation results, and $N^{recur}$ is the calculation results.

\begin{table}[htbp]
\caption{comparsion between simulation results of RM(128, 64) and PW(128, 64) with 500 realizations and the calculation results by the proposed recursive formula}
\begin{center}
\begin{tabular}{|c|c|c|c|c|c|}
\hline
\multicolumn{3}{|c|}{\textbf{RM(128,64)}}& \multicolumn{3}{|c|}{\textbf{PW(128,64)}} \\
\hline
$d$ & $N_d^{sim}$ & $N_d^{recur}$ & $d$ & $N_d^{sim}$ & $N_d^{recur}$ \\
\hline
16&2764.5&2766.9 &8&272.2&272  \\
\hline
18&397.1&393.5 &12&896.6&896  \\
\hline
20&80251 &80182&16&76812.2&77111  \\
\hline
\multicolumn{6}{l}{Note that $N_{10}=N_{14}=0$ for PW(128, 64)}
\end{tabular}
\label{tab2}
\end{center}
\end{table}

In PC-polar codes \cite{b4}, both reliability and code distance are taken into consideration when selecting the information set. A coefficient $\alpha$ is used to control the tradeoff between reliability and code distance. The larger $\alpha$ is, the greater code distance is. A parity check pattern can be considered as a realization of the pre-transformation matrix. Take PC-polar codes(128, 64) ($\alpha=1.5$) as an example, we calculate the average number of low-weight codewords. The result implies that pre-transformtion can increase the minimum code distance when the information set is properly chosen, that is, reducing the number of original minimum-codewords to zero. The number of low-weight codewords of the original code, a realization of the pre-transformed code and the code ensemble average are shown in Table.\;\ref{tab3}. In this case, although some rows of $F_N$ with Hamming weight 8 are selected into the information set, PC-polar codes can increase the minimum distance from 8 to 12. 

\begin{table}[htbp]
\caption{comparsion between the low-weight codewords number of the original codes, a realization of the pre-transformed codes and the ensemble average}
\begin{center}
\begin{tabular}{|c|c|c|c|}
\hline
Hamming weight &\multicolumn{3}{|c|}{$N_d$} \\
\cline{2-4}
 $d$ & Original & Pre-transformed & Average  \\
\hline
8  &32  &0 &0.5   \\
\hline
10  &0  &0  &0.0547  \\
\hline
12  &0  &48  &39.5   \\
\hline
14  &128  &28  &27   \\
\hline
16  &57048  &5228 &5250   \\
\hline
\end{tabular}
\label{tab3}
\end{center}
\end{table}  

In  CA-polar codes \cite{b3}, $r$ CRC bits are attached to $K$ information bits and all the $K'=K+r$ bits are fed into the polar encoder. To construct CA-polar codes, $K'$ indices are selected, and the first $K$ of them are information bits, while the others are dynamic frozen bits\cite{b19}. We construct  CA-polar code ($N=128, K=64, r=6$) by reliability sequence in \cite{b21}. The number of low-weight codewords of the original code, a CA-polar code with generator polynomial $g(D) = D^6+D+1$ and the code ensemble average are shown in Table.\;\ref{tab4}. 

\begin{table}[htbp]
\caption{comparsion between the low-weight codewords number of the original codes, a CA-polar code and the ensemble average}
\begin{center}
\begin{tabular}{|c|c|c|c|}
\hline
Hamming weight &\multicolumn{3}{|c|}{$N_d$} \\
\cline{2-4}
 $d$  & Original & CA & Average \\
\hline
8  &529  &4  &10.75  \\
\hline
10  &0  &0  &0.0547   \\
\hline
12  &0  &145  &85.5   \\
\hline
14  &0  &0  &27.07   \\
\hline
16  &3.364$\times 10^5$  &12550 &4952.4   \\
\hline
\end{tabular}
\label{tab4}
\end{center}
\end{table}  

Fig.\;\ref{fig4} and Fig.\;\ref{fig5} provide the BLER performances of various constructions under different list sizes, with reference to finite-length performance bounds such as normal approximation (NA), random-coding union (RCU) and meta-converse (MC) bounds \cite{b7} \cite{b16} \cite{b17}. PW pre, PC-Polar ($\alpha$ = 1.5), PC-Polar ($\alpha$ = 3.5) codes are specific realizations drawn from the code ensemble with different information set selections. The information sets of PAC and PC-Polar ($\alpha$ = 3.5) codes turn to be the same. In PAC codes, the transformation matrix $T$ is a upper-triangular Toeplitz matrix, while in PC-Polar codes ($\alpha$ = 3.5), $T$ is a randomly generated upper-triangular matrix. It is observed that reliability is the only contributing factor to decoding performance under SC decoding. Under SCL decoding with list size $L = 8$, the PC-polar codes ($\alpha = 1.5$) strike a good balance between reliability and distance, and shows the best decoding performance. When the list size is large enough, both PAC and PC-polar codes ($\alpha$ = 3.5) can approach NA bound with their ML performances.
\begin{figure}[htbp]
\centerline{\includegraphics[width = .5\textwidth]{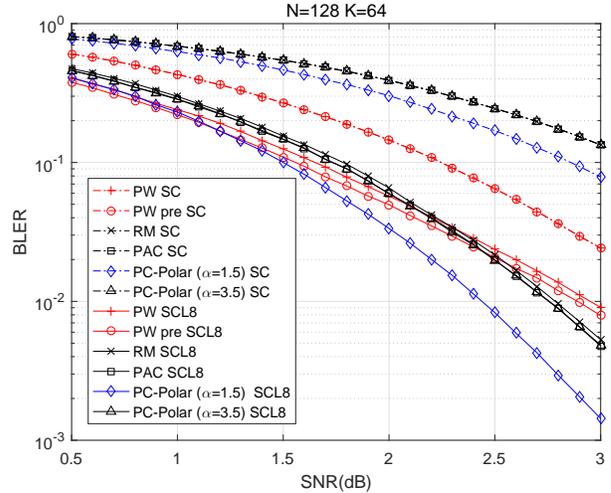}}
\caption{Performance comparison between different code constructions under SC and SCL, $L$=8}
\label{fig4}
\end{figure}        

\begin{figure}[htbp]
\centerline{\includegraphics[width = .5\textwidth]{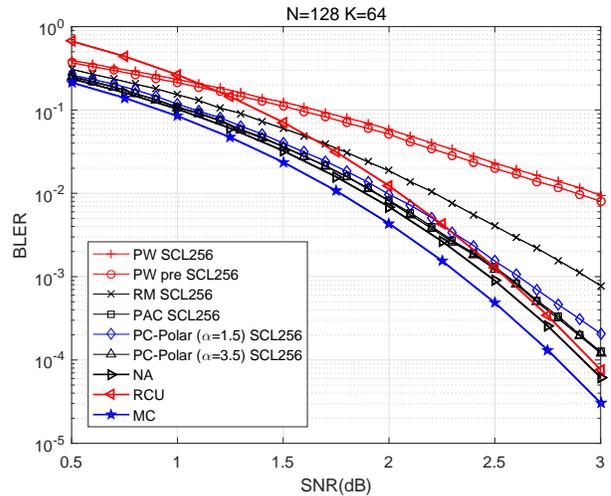}}
\caption{Performance comparison between different code constructions under SCL L=256, with reference to finite-length performance bounds}
\label{fig5}
\end{figure}

\section{Conclusion}
In this paper, we propose recursive formulas to efficiently calculate the average weight spectrum of pre-transformed polar codes, which include CA-polar, PC-polar and PAC codes as special cases. It is worth mentioning that our formulas work for any sub-channel selection criteria. We found that, with pre-transformation, the average number of minimum codewords decreases significantly, therefore outperforming the original RM/polar codes under the ML decoding and SCL decoding with large list sizes. Furthermore, as in the instance of PC-polar codes ($\alpha=1.5$), the combination of a proper sub-channel selection and pre-transformation has the potential to increase minimum code distance by eliminating minimum-weight codewords.

\newpage
\begin{appendix}
\emph{A. Proof of Theorem 1}
\begin{proof}
A trivial examination can prove the correctness of the boundary conditions. Let us focus on deriving the recursive formula.

\emph{Case 1}: $1 \leq i \leq 2^{m-1}$\\
$$g_{2^m}^{(i)}=f_{2^m}^{(i)} \oplus \sum_{j=i+1}^{2^{m-1}}T_{ij}f_{2^m}^{(j)} \oplus \sum_{j=2^{m-1}+1}^{2^m}T_{ij}f_{2^m}^{(j)}.$$
Let $f_{2^m}^{(i)} \oplus \sum\limits_{j=i+1}^{2^{m-1}}T_{ij}f_{2^m}^{(j)} \triangleq [\textbf{X},\textbf{0}]$, $\sum\limits_{j=2^{m-1}+1}^{2^m}T_{ij}f_{2^m}^{(j)} \triangleq [\textbf{Y},\textbf{Y}]$, where $\textbf{0}$ is an all-zero row vector of length $ 2^{m-1}$, $\textbf{X}=\left(x_1,\dots,x_{2^{m-1}}\right)$ , $\textbf{Y}=\left(y_1,\dots,y_{2^{m-1}}\right)$.

Apparently, \textbf{X} and \textbf{Y} are independent, and ${\forall} \ \textbf{a}=(a_1,\dots,a_{2^{m-1}}) \in \{0,1\}^{2^{m-1}}$, $P(\textbf{Y}=\textbf{a})=2^{-2^{m-1}}$. Let $w(\textbf{X})=d_1$, $w(\textbf{Y})=d_2$, and $c$ be the number of positions where \textbf{X} and \textbf{Y} are both 1. We have
\begin{align*}
w(g_{2^m}^{(i)})&=w([\textbf{X} \oplus \textbf{Y},\textbf{Y}])\\
&=w(\textbf{X} \oplus \textbf{Y})+w(\textbf{Y})\\
&=d_1+2d_2-2c.
\end{align*}

Because $d_2 \geq c$ and $d_1 \geq w(f_{2^m}^{(i)})$ \cite[Corollary 1]{b9}, the equation $w(g_{2^m}^{(i)})=d_1+2d_2-2c=w(f_{2^m}^{(i)})$ holds if and only if $d_1=w(f_{2^m}^{(i)})$, $d_2=c$. In fact, \{$d_2=c$\} denotes the event that \textbf{X} \emph {covers} the first half of \textbf{Y}, i.e., if $x_i=0$ then $y_i =0$, for all $1 \leq i \leq 2^{m-1}$. Let $\{i_1,\dots,i_{d_1}\}$ denote the $d_1$ locations where $x_{i_1},\dots,x_{i_{d_1}}=1$, hence the recursive formula is
\begin{align*}
P(m,i)&=P(d_1=w(f_{2^m}^{(i)}))*P(d_2=c|d_1=w(f_{2^m}^{(i)}))\\
&=P(m-1,i)*\\
&P\left(y_{i_1},\dots,y_{i_{d_1}} \in \{0,1\}^{d_1},y_i=0 \ otherwise \right)\\
&=P(m-1,i)*\frac{2^{d_1}}{2^{2^{m-1}}}\\
&=P(m-1,i)*\frac{2^{w(f_{2^m}^{(i)})}}{2^{2^{m-1}}}.
\end{align*}

\emph{Case 2}: $2^{m-1} < i \leq 2^m$
\begin{align}
\label{eq:7}
g_{2^m}^{(i)}&=\left[f_{2^m}^{(i)} \oplus \sum_{j=i+1}^{2^m}T_{ij}f_{2^m}^{(j)} \right] \notag \\
&=\bigg[f_{2^{m-1}}^{(i-2^{m-1})} \oplus  \sum_{j=i+1}^{2^m}T_{ij}f_{2^{m-1}}^{(j-2^{m-1})}, \notag \\
& \ \ \ \ f_{2^{m-1}}^{(i-2^{m-1})} \oplus \sum_{j=i+1}^{2^m}T_{ij}f_{2^{m-1}}^{(j-2^{m-1})} \bigg] \notag \\
&\backsim \left[g_{2^{m-1}}^{(i-2^{m-1})},g_{2^{m-1}}^{(i-2^{m-1})}\right],
\end{align}
where $X_1\backsim X_2$ means $X_1,X_2$ have the same distribution.
\end{proof}

\emph{B. Proof of Theorem 2}
\begin{proof}
(\ref{eq:6}) is obtained with the observation that $w\left(f_N^{(1)}\right)$ is odd and $\forall$ $i > 1$, $w\left(f_N^{(i)}\right)$ is even.

\emph{Case 1}: $1 \leq i \leq 2^{m-1}$

Similar to the proof of \textbf{Theorem 1}, let $w(\textbf{X})=d_1$, $w(\textbf{Y})=d_2$ and $c$ be the number of positions where \textbf{X} and \textbf{Y} are both 1. Denoted by $\mathcal{V}=\{v_1,\dots,v_c\}$  the set of positions where \textbf{X} and \textbf{Y} are both 1, and $\mathcal{V}^c$  its complement. Let $\textbf{Y}_{\mathcal{V}^c}$\ denote the corresponding subvector of \textbf{Y}, we have $w\left(\textbf{Y}_{\mathcal{V}^c}\right)=d_2-c$. Because
\begin{align*}
w(g_{2^m}^{(i)})&=w([\textbf{X} \oplus \textbf{Y},\textbf{Y}])\\
&=w(\textbf{X} \oplus \textbf{Y})+w(\textbf{Y})\\
&=d_1+2d_2-2c\\
&=d,
\end{align*}
then $d_2-c=\frac{d-d_1}{2}$, so $d-d_1$ must be even. No matter what $c$ is, the equation is satisfied if and only if $w\left(\textbf{Y}_{\mathcal{V}^c}\right)=\frac{d-d_1}{2}$. Based on the above observations, $P(m,i,d)$ can be formulated as
\begin{align*}
&P(m,i,d)\\
&=\sum\limits_{\mbox{\tiny$\begin{array}{c}
d'=w\left(f_{2^m}^{(i)}\right)\\
d-d' is \ even \end{array}$}}^{d}P(m,i,d|w\left(\textbf{X}\right)=d') \ast P(w\left(\textbf{X}\right)=d')\\
&=\sum\limits_{\mbox{\tiny$\begin{array}{c}
d'=w\left(f_{2^m}^{(i)}\right)\\
d-d' is \ even \end{array}$}}^{d}P(m,i,d|w\left(\textbf{X}\right)=d') \ast P(m-1,i,d').
\end{align*}
The last equality holds due to $\textbf{X}\backsim g_{2^{m-1}}^{(i)}$.

In particular
\begin{align*}
P(m,i,d|w\left(\textbf{X}\right)=d') &= P\left(w\left(\textbf{Y}_{\mathcal{V}^c}\right)=\frac{d-d_1}{2}\right)\\
&= \frac{2^{d'}\begin{pmatrix} 2^{m-1}-d' \\ \frac{d-d'}{2} \end{pmatrix}}{2^{2^{m-1}}}.
\end{align*}
Consequently, the recursive formula is
\begin{align*}
P(m,i,d)=\sum\limits_{\mbox{\tiny$\begin{array}{c}
d'=w\left(f_{2^m}^{(i)}\right)\\
d-d' is \ even \end{array}$}}^{d}P(m-1,i,d')\frac{2^{d'}\begin{pmatrix} 2^{m-1}-d' \\ \frac{d-d'}{2} \end{pmatrix}}{2^{2^{m-1}}}.
\end{align*}
\emph{Case 2}: $2^{m-1} < i \leq 2^m$, according to (\ref{eq:7})
\begin{align*}
g_{2^m}^{(i)}
&\backsim \left[g_{2^{m-1}}^{i-2^{m-1}},g_{2^{m-1}}^{i-2^{m-1}}\right].
\end{align*}
It is straightforward to obtain the recursive formula
$$P(m,i,d)=P(m-1,i-2^{m-1},d/2).$$

\end{proof}

\end{appendix}

%%%%%%
%% To balance the columns at the last page of the paper use this
%% command:
%%
%\enlargethispage{-1.2cm} 
%%
%% If the balancing should occur in the middle of the references, use
%% the following trigger:
%%
%\IEEEtriggeratref{4}
%%
%% which triggers a \newpage (i.e., new column) just before the given
%% reference number. Note that you need to adapt this if you modify
%% the paper.  The "triggered" command can be changed if desired:
%%
%\IEEEtriggercmd{\enlargethispage{-20cm}}
%%
%%%%%%

%%%%%%
%% References:
%% We recommend the usage of BibTeX:
%%
%\bibliographystyle{IEEEtran}
%\bibliography{definitions,bibliofile}
%%
%% where we here have assume the existence of the files
%% definitions.bib and bibliofile.bib.
%% BibTeX documentation can be obtained at:
%% http://www.ctan.org/tex-archive/biblio/bibtex/contrib/doc/
%%%%%%

%% Or you use manual references (pay attention to consistency and the
%% formatting style!):

\end{document}